%% file: main.tex
\documentclass[11pt]{article}

\usepackage{macros}

\newif\ifnotes

\title{Explicit Almost-Optimal $\eps$-Balanced Codes \\
via Free Expander Walks}

\author{Jun-Ting Hsieh\thanks{MIT. \texttt{juntingh@mit.edu}.}
\and Sidhanth Mohanty\thanks{Northwestern University. \texttt{sidhanth.mohanty@northwestern.edu}.}
\and Rachel Yun Zhang\thanks{MIT.  \texttt{rachelyz@mit.edu}. Supported by the Schwarzman College of Computing Future Research Fellowship supported by Google.}}

\date{\today}

\begin{document}

%% TITLE AND ABSTRACT
\sloppy
\maketitle
\begin{abstract}
\input{abstract}
\end{abstract}

\thispagestyle{empty}
\setcounter{page}{0}

\newpage

% TABLE OF CONTENTS
% \enlargethispage{1cm}
% \tableofcontents
% \pagenumbering{roman}
% \newpage
% \pagenumbering{arabic}

% \input{notation}

\input{intro}

\input{ow-based-construction}

\input{free-expanders}

% \input{pro-version}

%% BIBLIOGRAPHY

\bibliographystyle{alpha}
\bibliography{refs}

\end{document}

%% file: abstract.tex
We study the problem of constructing explicit binary codes whose rate and distance match the Gilbert--Varshamov bound in the low-rate, high-distance regime.
In 2017, Ta-Shma \cite{TaS17} gave an explicit family of codes where every pair of codewords has relative distance $\frac{1-\varepsilon}{2}$, with rate $\Omega(\varepsilon^{2+o(1)})$, matching the Gilbert--Varshamov bound up to a factor of $\varepsilon^{o(1)}$.
Ta-Shma's construction was based on starting with a good code and amplifying its bias with walks arising from the $s$-wide-replacement product.

In this work, we give a simpler almost-optimal construction, based on what we call \emph{free expander walks}---ordinary expander walks where each step is taken on a distinct expander from a carefully chosen sequence.
This sequence of expanders is derived from the construction of near-$X$-Ramanujan graphs due to O'Donnell and Wu \cite{OW20}. We additionally discuss some additional applications of near-$X$-Ramanujan graphs to ``on average'' lossless expansion and rotating expanders.

%% file: intro.tex
\section{Introduction}

A longstanding and fundamental open problem in coding theory is to construct binary codes with the best possible tradeoffs between rate and distance. 
% The classical Gilbert-Varshamov bound guarantees for all $\delta > 0$ the existence of a binary code family with relative distance $\delta$ and rate $\ge 1 - h(\delta) - o(1)$, where $h$ is the binary entropy function. Finding an \emph{explicit} construction of such codes achieving this rate-distance tradeoff is an important open question in coding theory. For codes with distance close to $\frac12$, the GV bound says that a random code with rate $\eps^2$ also has relative distance $\ge \frac12 - \eps$. 
In this work, we seek to understand the best possible rate achievable by a code family whose relative distance approaches $\frac12$. The classical Gilbert--Varshamov (GV) bound guarantees the existence of a binary code family with relative distance $\ge \frac12 - \eps$ and rate $\Omega(\eps^2)$ --- indeed, a random linear code satisfies such a property; this is tight up to a $\log\frac{1}{\eps}$ factor due to the lower bound of \cite{MRRW77} (see also \cite{Alo09}). However, constructing \emph{explicit} codes meeting this rate-distance tradeoff is not yet known and represents an important open problem.
% In many parameter regimes, the best tradeoff for which binary codes are known to exist is given by the Gilbert-Varshamov (GV) bound, which is saturated by a random linear code. However, we do not yet know how to explicitly construct binary codes meeting the GV bound. Constructing such codes remains an important question in coding theory and pseudorandomness.

% For codes with distance $\ge \frac12 - \eps$, the GV bound says that the best achievable rate is $\Omega(\eps^2)$. In this work, we study the problem of constructing explicit codes meeting the GV bound in this high distance regime. 

% An \emph{$\eps$-balanced code} is a linear code where the relative distance between any pair of codewords is in $\bracks*{ \frac{1-\eps}{2}, \frac{1+\eps}{2} }$.
% An important problem in coding theory and pseudorandomness is to understand the best possible rate achievable for an $\eps$-balanced code.
% The classical Gilbert--Varshamov (GV) bound guarantees the existence of an $\eps$-balanced code with rate $\Omega(\eps^2)$---indeed, a random linear code satisfies such a property; this is tight up to a $\log\frac{1}{\eps}$ factor due to the lower bound of \cite{MRRW77} (see also \cite{Alo09}).

% \parhead{Explicit constructions.}
Following~\cite{TaS17}, we seek to construct what are known as \emph{$\eps$-balanced codes}, where the relative distance between \emph{any} pair of codewords is in the interval $\bracks*{ \frac{1-\eps}{2}, \frac{1+\eps}{2} }$. Notice that the property of being $\eps$-balanced is more restrictive than having relative distance $\ge \frac12 - \eps$, which simply requires that the relative distance between every pair of codewords is bounded below by $ \frac12 - \eps$. 

In the context of error-correction, explicit constructions of $\eps$-balanced codes often have structure that plays well with designing fast decoding algorithms \cite{AJQST20,JQST20,JST21}, a feature that random linear codes do not enjoy.
In pseudorandomness, constructions of $\eps$-balanced codes immediately give constructions of \emph{almost $k$-wise uniform distribution} \cite{NN90}, which is an important primitive in a myriad of applications: derandomization \cite{NN90}, probabilistically checkable proofs \cite{BGHSV04}, constructing PRGs for restricted classes of functions \cite{GMRTV12}, expander constructions \cite{MOP20,OW20}, and constructing explicit quantum codes \cite{JMOPT22} to name a few.

Obtaining explicit constructions of $\eps$-balanced codes achieving the GV bound is an active thread in coding and pseudorandomness. 
One popular approach for constructing good codes is to \emph{concatenate} a large alphabet code with a small binary code~\cite{Justensen72}.
In the $\frac{1}{2}-\eps$ relative distance regime, explicit codes obtained via concatenation achieve rate $\Omega(\eps^3)$, a result traditionally attributed to Zyablov \cite{Zya71}; see \cite[Exercise 10.3]{GRS12} for an English presentation.
A second approach for obtaining good $\eps$-balanced codes, pioneered by Rozenman and Wigderson and exposited in the lecture notes of Bogdanov \cite{Bog12}, is to amplify bias via expander walks. In the original presentation, such expander walk-based codes obtained only a rate of $\Omega(\eps^4)$. In 2017, the breakthrough work of Ta-Shma~\cite{TaS17} showed how to sparsify these expander walks using a wide replacement product, thereby obtaining an explicit family of $\eps$-balanced codes with rate $\Omega\left( \eps^{2+o(1)} \right)$, matching the GV bound up to a subpolynomial factor in $1/\eps$. The subpolynomial factor was recently improved in~\cite{CohenC25}. 

% \emph{Expander-based amplification methods} for obtaining explicit constructions were pioneered by Rozenman and Wigderson, exposited in lecture notes of Bogdanov \cite{Bog12}.The original form of these methods obtained an explicit construction with rate $\Omega(\eps^4)$. In 2017, the breakthrough work of~\cite{TaS17} sparsified the The current state of the art is given by the breakthrough work of~\cite{TaS17}, who constructs an explicit family of $\eps$-balanced codes with rate $\Omega \left( \eps^{2+o(1)} \right)$, matching the GV bound up to a subpolynomial factor in $1/\eps$, which was reduced in recent work~\cite{CohenC25}. The code family is constructing using the \emph{wide replacement walk}, which is a specific sparsification of an expander walk obtained by using the randomness from a smaller expander to direct the walk. These constructions 
% Concretely, an explicit construction of an $n$-bit $\eps$-balanced code $\calC$ with rate $r$ refers to a deterministic $\poly(n)$ time algorithm that outputs an $n\times rn$ generator matrix of the code.

In this work, we present what we believe is a much simpler explicit construction of a nearly optimal $\eps$-balanced code with rate $\Omega(\eps^{2+o(1)})$.

\begin{theorem} \label{thm:main}
    There is an explicit $\eps$-balanced binary linear code $\calC\subseteq\{ 0, 1 \}^n$ with rate $\Omega\parens*{ \eps^{2+o(1)}  }$.
\end{theorem}

Our code family is obtained by taking the simple expander walk, with a different expander at each step. While we obtain a worse subpolynomial error term than that achieved by~\cite{TaS17,CohenC25}, we hope that our simplification can lead to progress in both explicit constructions and algorithms.

\begin{remark}
    The work of Jeronimo, Mittal, Roy, and Wigderson \cite{JMRW25} proved that Ta-Shma's construction extends to amplifying bias in the \emph{operator} setting as well, a fact useful in constructing almost-optimal expanders on groups.
    In contrast, it is possible to construct operators for which our construction fails to achieve almost-optimal amplification.
\end{remark}

% \subsection{Our results}
% In this work, we give an alternate construction of nearly optimal $\eps$-balanced codes that we believe is simpler to state and analyze than that of Ta-Shma \cite{TaS17}.
% \begin{theorem} \label{thm:main}
%     There is an explicit $\eps$-balanced binary linear code $\calC\subseteq\{ 0, 1 \}^n$ with rate $\Omega\parens*{ \eps^{2+o(1)}  }$.
% \end{theorem}

% \begin{remark}
%     In this paper, we will find it more convenient to present all our definitions and proofs for codes in the $\{\pm1\}$-basis, instead of the $\{0,1\}$-basis.
%     In this $\{\pm1\}$-basis, binary linear codes correspond to \emph{multiplicative} subspaces of $\{\pm1\}^n$.
%     One can convert from the $\{0,1\}$-basis (and back) via the map $x \longleftrightarrow ( (-1)^{x_i} )_{i \in [|x|]}$.
% \end{remark}

\subsection{Overview}
    
To motivate our construction, we begin by discussing a construction based on \emph{expander walks} that achieves rate $\Omega(\eps^4)$, which is also the starting point for Ta-Shma's construction.

Given a code $\calC_0 \subseteq \F_2^{n_0}$ with bias $\eps_0$, rate $r_0$, and block length $n_0$, we may consider the code $\calC$ obtained as an \emph{expander walk lift} with block length $n_0 \cdot d^{\ell}$ obtained by taking a $d$-regular $\lambda$-spectral expander $G$, and defining the codewords to be $f(x)_{i_0,\dots,i_{\ell}}\coloneqq x_{i_0} + \dots + x_{i_{\ell}}$ where $x \in \calC_0$ and $i_0,\dots,i_{\ell}$ is a length-$\ell$ walk in $G$.

The bias can be written as $\angles*{\mathbf{1}, D_x (A_G D_x)^{\ell} \mathbf{1} }$, where $A_G$ is the normalized adjacency matrix of $G$, $\mathbf{1}$ is the all-1s vector, $D_x$ is an $n_0\times n_0$ diagonal matrix where the $(i,i)$-th entry is $(-1)^{x_i}$, and $\angles*{\cdot,\cdot}$ is the normalized inner product.
One may use spectral properties of $A_G$ to prove that $\norm*{(A_G D_x)^2}_{\op} \le \lambda$, which in turn implies that the bias is at most $\lambda^{\ell/2}$.
Since we can choose $\lambda$ to be $\Theta\parens*{\frac{1}{\sqrt{d}}}$, one obtains a rate of $\Omega(\eps^4)$.

Another way to visualize the effect of the walk $(A_GD_x)^\ell$ is via the following schematic, in which arrows represent how mass may move along the $\mathbf{1}$ and $\perp$ directions (orthogonal to $\mathbf{1}$) by applications of $D_x$ and $A_G$. The unlabeled black arrows represent an upper bound of $1$, i.e., it's possible that an application of $D_x$ or $A_G$ does not shrink the norm of the vector. The norm of the vector obtained by successive applications of $A_G D_x$ can be thought of as the maximum over all paths through the diagram of the product of all of the factors it passes through. 

\begin{figure}[h!]
    \centering
    \includegraphics[width=0.8\linewidth]{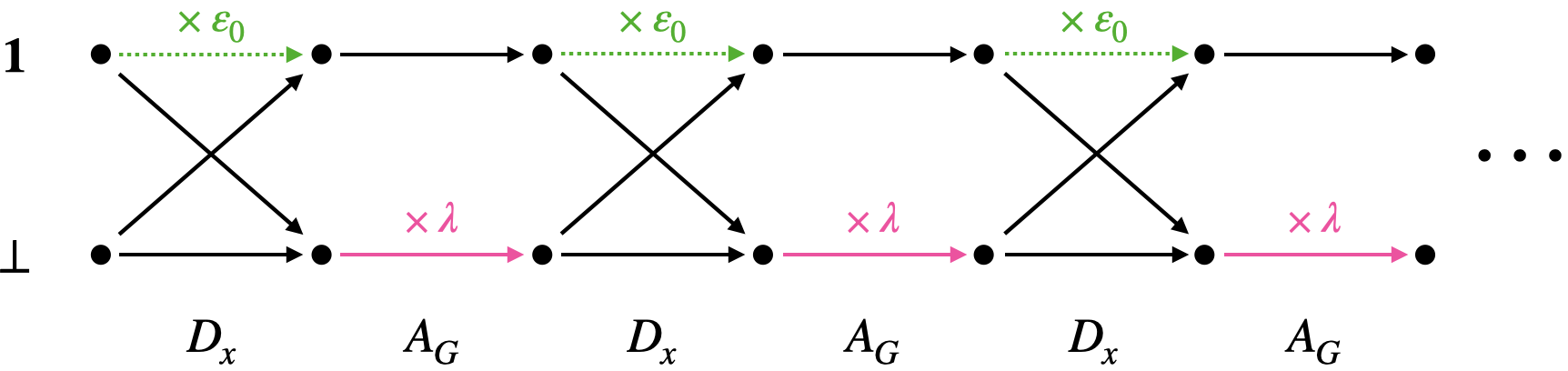}
    % \caption{A schematic of the walk $(A_G D_x)^\ell$. Arrows represent upper bounds on the mass moved between the $\mathbf{1}$ and $\perp$ directions by applications of $D_x$ and $A_G$. The unlabeled black arrows represent an upper bound of $\times 1$.}
    \label{fig:flow-digram}
\end{figure}

Consider the following highlighted path that switches between the $\mathbf{1}$ and $\perp$ components on every application of $D_x$. This path only picks up a $\lambda$ factor in every other application of $A_G$, explaining why we may only experience a $\lambda^{\ell/2}$ decay over $\ell$ steps.

\begin{figure}[h!]
    \centering
    \includegraphics[width=0.8\linewidth]{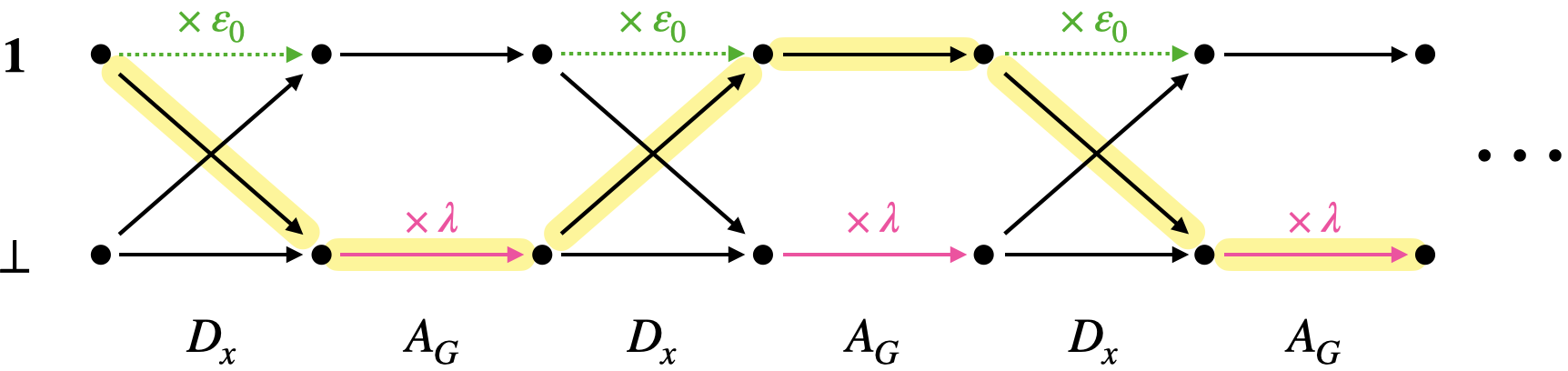}
    \label{fig:flow-digram-highlight}
\end{figure}

It is natural to wonder whether the bound of $\lambda$ for two steps of the walk can be improved to a $\lambda^2$.
Indeed, such an improvement would immediately imply \Cref{thm:main}.
Unfortunately, it is possible to instantiate a graph $G$ along with $x$ such that the bias is genuinely as large as $\lambda^{\ell/2}$. Consider for instance $x$ that is the second eigenvector of $A_G$: the highlighted path precisely describes the evolution of successive applications of $A_G D_x$.
%\footnote{It is possible to construct $G$ for which the second eigenvector is a balanced Boolean vector.}

% \begin{example} \label{ex:bad-example}
%     It is possible to obtain $G$ as a lift of the $2$-vertex $d$-regular graph with $\approx \frac{1-\lambda}{2}\cdot d$ parallel edges between the two vertices, $\approx \frac{1+\lambda}{2}\cdot d$ self-loops per vertex, and $x$ aligned with the vertices of the base graph \cite{Bor19}.
%     The sequence $x_{i_0},\dots,x_{i_{\ell}}$ behaves as if it were generated from a $2$-state Markov chain on $\{\pm1\}$ in which
%     \[
%         x_{i+1} =
%         \begin{cases}
%             x_i &\text{with probability }\frac{1+\lambda}{2} \\
%             -x_i &\text{with probability }\frac{1+\lambda}{2} \mper
%         \end{cases}
%     \]
%     Then the bias can be expressed as:
%     \[
%         \E x_{i_0}\cdots x_{i_{\ell-1}} x_{i_{\ell}} = \lambda \E x_{i_0}\cdots x_{i_{\ell-1}}^2 = \E x_{i_0} \cdots x_{i_{\ell-2}}\mper
%     \]
%     Recursing on the above gives a bias of $\lambda^{\ell/2}$.

%     If the number of self-loops were instead $\frac{1-\lambda}{2}\cdot d$, the bias would be $(-\lambda)^{\ell/2}$.
% \end{example}

\parhead{Free expander walks.}
% The takeaway from \Cref{ex:bad-example} is that 
In general, the $\lambda^{\ell/2}$ decay is realized if $x$ aligns with an atypically sparse or abnormally dense cut in $G$.
Our way around this is to take each step on a \emph{different} expander.
If we have an ensemble of expanders $H_1,\dots,H_t$ that ``look pseudorandom'' to each other, then no fixed $x$ can align with an abnormal cut in many of the $H_i$, and in fact $x$ should ``look random'' to most of them.
Concretely, we assume an ensemble of $\lambda$-spectral expanders $H_1,\dots,H_t$ satisfying the \emph{all-signings near-Ramanujan property}, i.e.
\[
    \max_{\sigma:[t]\to\{\pm1\}} \max_{y:\angles*{y,\mathbf{1}} = 0} \E_{\bi\sim[t]} \frac{\angles*{y, \sigma_{\bi} A_{H_{\bi}}y}}{\angles*{y,y}} \le \frac{\lambda}{\sqrt{t}}\mper
\]
Explicit constructions of such ensembles are implied by the work of O'Donnell and Wu \cite{OW20}. 

Given such an ensemble of expanders, it is not too hard to show that for most graphs $H_i$, most of the mass in the $\mathbf{1}$ direction cannot return to the $\mathbf{1}$ direction after just one step according to $H_i$. More generally, we show that even for larger numbers of steps, most sequences $H_{i_1}, \dots, H_{i_\arr}$ will not cause most of the mass to return to the $\mathbf{1}$ direction. Precisely, we show that 
% Our key technical insight, articulated in \Cref{lem:random-sequence-good}, is that this property implies that 
for any given $x$, a random word $\bw$ in $[t]^{\arr}$ of not-too-large length $\arr$ satisfies $\norm*{ A_{H_{\bw_\arr}} D_x \cdots A_{H_{\bw_1}} D_x }_{\op} \le \lambda^{\arr-1}$ with high probability.
Thus, if we perform an expander walk whose schedule includes every word in $[t]^{\arr}$, then we obtain an almost-optimal spectral norm bound on most words. 

\subsection{Organization}
In \Cref{sec:construction-and-analysis}, we present the construction based on ensembles satisfying the all-signings near-Ramanujan property, along with the proof of \Cref{thm:main}.

In \Cref{sec:free-expanders}, we give a brief description of the construction of ``free expanders'' by O'Donnell and Wu \cite{OW20} and how the ensembles we use arise as a special case.
We also show a few other properties of these ensembles and some applications of their result, including ``on-average'' lossless expansion and rotating expanders.
These applications are of independent interest, and we hope that they draw further attention to these objects.

%% file: ow-based-construction.tex
\section{Almost-Optimal \texorpdfstring{$\eps$}{epsilon}-Balanced Codes}
\label{sec:construction-and-analysis}

\subsection{Ingredients and Notation}

% \begin{definition}[Binary linear code]
%     A \emph{binary linear code} is a subspace $\mathcal{C} \subseteq \{ 0, 1 \}^n$. 
% \end{definition}

In this paper, we find it more convenient to present our definitions and proofs for binary codes in the $\{\pm1\}$-basis, instead of the $\{0,1\}$-basis.
In this $\{\pm1\}$-basis, binary linear codes correspond to \emph{multiplicative} subspaces of $\{\pm1\}^n$.
One can convert from the $\{0,1\}$-basis (and back) via the map $x \in \{0,1\}^n \longleftrightarrow ( (-1)^{x_i} )_{i \in [n]} \in \{\pm1\}^n$.

% In this paper, we will find it more convenient to work in the $\{ \pm 1 \}$ basis, instead of the $\{ 0, 1 \}$ basis. In the $\{ \pm 1 \}$ basis, binary linear codes correspond to \emph{multiplicative} subspaces of $\{ \pm 1 \}^n$. One can convert from the $\{ 0, 1 \}$ basis to the $\{ \pm 1 \}$ basis (and back) via the conversion $x \longleftrightarrow ( (-1)^{x_i} )_{i \in [|x|]}$. As such, we find it more convenient to present all our definitions and proofs for codes in the $\{ \pm 1 \}$ basis.

\begin{definition}[Binary codes]
    A \emph{binary code} is a subset $\calC \subseteq \{ \pm 1 \}^n$. It is said to be \emph{linear} if it is in fact a multiplicative subspace of $\{ \pm 1 \}^n$. Its \emph{dimension} is $\log_2 |\calC|$, and its \emph{rate} is $\frac{\log_2 |\calC|}{n}$. 
\end{definition}

We will be interested in binary linear codes with relative distance between $\frac{1-\eps}{2}$ and $\frac{1+\eps}{2}$. In the multiplicative notation, this corresponds to the code having \emph{bias} $\le \eps$.

\begin{definition}[Bias]
    We will say that a string $x \in \{ \pm 1 \}^n$ has \emph{bias} $\eps$ if $\bias(x) \coloneqq |\bbE_{i \in [n]} x_i| = \eps$.
    We say that a binary linear code $\calC \subseteq \{ \pm 1 \}^n$ has bias $\le \eps$ if for all $x \in \calC \setminus \{ 1^{n} \}$, $\bias(x) \le \eps$.
\end{definition}

\parhead{Base code.}
Let $\calC_0\subseteq\{ \pm 1\}^{n_0}$ be an explicit linear code of block length $n_0$, bias $\le \eps_0 \in (0, 1/4]$, and rate $r_0 = \poly(\eps_0) > 0$.
For concreteness, we can take $\calC_0$ to be the construction of \cite{ABNNR02}. We will later choose $\eps_0$ to depend on our target bias $\eps$.

\parhead{Graphs, matrices, and inner product.}
For a graph $G$, we use $A_G$ to denote its normalized adjacency matrix.
We use $\mathbf{1}$ to refer to the all-ones vector.
For $x\in\{\pm1\}^{n_0}$ we set $D_x=\diag(x)$.
Throughout, we use the \emph{normalized inner product}, for functions $f,g:[n_0]\to\R$, we define $\angles*{f, g} \coloneqq \E_{\bv\sim[n_0]} f(\bv) \cdot g(\bv)$, and write $\norm*{f} \coloneqq \sqrt{\angles*{f, f}}$ for the induced norm.

\parhead{All-signings near-Ramanujan families.}
We now describe the kind of expander family we use for defining our expander walk.
\begin{definition} \label{def:all-signings-expanders}
    We say that a collection of $d$-regular graphs $(H_j)_{j\in[t]}$ on the same vertex set $V$ is an \emph{all-signings near-Ramanujan family} if for every signing $\sigma\in\{\pm1\}^t$, the signed $td$-regular multigraph
    \[
        H(\sigma) \coloneqq \sum_{j=1}^t \sigma_j H_j
    \]
    satisfies $\norm*{ A_{H(\sigma)}|_{\mathbf{1}^{\perp}} }_{\op} \le \frac{2}{\sqrt{td}}$, 
    and further, $\norm*{A_{H_j}|_{\mathbf{1}^{\perp}}}_{\op} \le \frac{2}{\sqrt{d}}$ for all $j\in[t]$.
\end{definition}

Note that a Ramanujan $d$-regular graph $H$ has $\norm*{A_H|_{\mathbf{1}^{\perp}}}_{\op} \leq \frac{2\sqrt{d-1}}{d}$, which is smaller than $\frac{2}{\sqrt{d}}$.
In \Cref{def:all-signings-expanders}, We use the bounds $\frac{2}{\sqrt{td}}$ and $\frac{2}{\sqrt{d}}$ because we only have constructions of near-Ramanujan families (and also for simplicity).

\begin{theorem}[Special case of {\cite[Theorem 10.13]{OW20}}] 
\label{thm:OW} 
    For every $t \ge 1$, and $d\ge 2$, there is an explicit all-signings near-Ramanujan family $\{ G_{n_0} \}_{n_0 \in \bbN}$ where $G_{n_0}$ is on a vertex set $V_{n_0}$ such that $|V_{n_0}| = n_0 + o(n_0)$.
\end{theorem}

\begin{remark}
    For a reader interested in understanding the proof of \Cref{thm:OW}, the simplest self-contained route is to check that the construction of Mohanty, O'Donnell, and Paredes \cite{MOP20} can be adapted to satisfy this property.
    Here we simply appeal to the far-reaching generalization of \cite{MOP20} given by O'Donnell and Wu \cite[Theorem 10.13]{OW20}. % is extremely general: any noncommutative self-adjoint polynomial (possibly with matrix coefficients) of matchings and permutations have all their nontrivial eigenvalues approximately contained within the spectrum of the corresponding universal cover.
\end{remark}

\begin{remark}
    For a fixed signing, the signed sum of a collection of independent random graphs is near-Ramanujan with high probability.
    Since there are only constantly many signings, by the union bound, a random ensemble is all-signings near-Ramanujan.
    The construction of \cite{OW20} is obtained by derandomizing the randomized construction, and hence it is natural that it inherits this basic property.
\end{remark}

\subsection{The Construction}
\label{sec:construction}

Fix integers $t\ge 2$ and $d$.
Using \Cref{thm:OW}, we construct an all-signings near-Ramanujan family of $d$-regular graphs $H_1,\dots,H_t$ on $[n_0]$.\footnote{Strictly speaking, for a given $n_0$, we have $|V(H_1,\dots,H_t)| = n_0+o(n_0)$.
However, we can pad the base code with zeros if necessary, so the block length is equal to the number of vertices and incur only a $1+o(1)$ blowup in block length.}
The following definition formally articulates the notion of a free expander walk.
\begin{definition}[Expander walk lift]
    Given a word $W\in[t]^{\ell}$ and $x\in\{\pm1\}^{n_0}$, we define its \emph{$W$-expander walk lift} $f_W(x)\in\{\pm1\}^{n_0\cdot d^{\ell}}$ as follows.
    Let $I_W$ denote the set of all length-$\ell$ walks $(i_0,i_1,\dots,i_{\ell})\in[n_0]^{\ell+1}$, where $i_0\in[n_0]$ and $i_s$ is a neighbor of $i_{s-1}$ in $H_{W[s]}$ for all $s\in[\ell]$.
    The coordinates of $f_W(x)$ are indexed by walks $i = (i_0,i_1,\dots,i_{\ell})\in I_W$, and we set
    \[
        f_W(x)_{i} = x_{i_0} \cdot x_{i_1} \cdots x_{i_{\ell}}\mper
    \]
\end{definition}

For example, a word $W = (1,2,3)$ gives rise to $n_0 d^{3}$ length-3 walks ($i_0,i_1, i_2, i_3)$, where $(i_0, i_1) \in E(H_1)$, $(i_1, i_2) \in E(H_2)$, and $(i_2, i_3) \in E(H_3)$.

We now describe the word $W$ that gives the sequence of expanders we use in our construction.
\begin{definition}[Expander schedule]
    For some integers $\arr$ and $r$, let $\ell \coloneqq \arr \cdot t^\arr \cdot r$ and construct $W^*\in[t]^\ell$ by concatenating all words in $[t]^\arr$ in an arbitrary order, and then repeating this $r$ times; i.e.,
    \[
        W^* = \underbrace{w^{(1)}_1\,\cdots\,w^{(1)}_\arr}_{w^{(1)}\in[t]^\arr}\;\;\Big\Vert\;\;\underbrace{w^{(2)}_1\,\cdots\,w^{(2)}_\arr}_{w^{(2)}\in[t]^\arr}\;\;\Big\Vert\;\;\cdots\;\;\Big\Vert\;\;\underbrace{w^{(r\cdot t^\arr)}_1\,\cdots\,w^{(r\cdot t^\arr)}_\arr}_{w^{(r\cdot t^\arr)}\in[t]^\arr} \,,
    \]
    so $W^*$ has length $\ell$.
\end{definition}

% For $x\in\calC_0 \subseteq \{ \pm 1 \}^{n_0}$, define its \emph{expander walk lift} $f(x)\in\{ \pm 1 \}^{n_0 d^\ell}$ as follows. Let $I$ denote the set of length $\ell$ walks $(i_0, i_1, \dots, i_\ell) \in [n_0]^{\ell+1}$, where $i_0 \in [n_0]$ and $i_s$ is a neighbor if $i_{s-1}$ in $G_s$ for all $s \in [\ell]$. The coordinates of $f(x)$ are indexed by walks $i = (i_0, \dots, i_\ell) \in I$. 
% Note that there are $n_0 \cdot d^\ell$ such walks, and thus $f(x)$ has length $n := |I| = n_0 \cdot d^\ell$.
% We set
% \[
%     f(x)_{(i_0, \dots, i_\ell)} \;=\; x_{i_0} \cdot x_{i_1}\cdots x_{i_\ell}.
% \]
\parhead{Construction of our code and parameter settings.}
Henceforth, we define $f\coloneqq f_{W^*}$, and define our code $\calC$ as the image $f(\calC_0) \subseteq  \{\pm1\}^{n_0 d^{\ell}}$.
Linearity of $\calC$ follows from multiplicativity of $f$.

We will set $t = d^2$, and define $\lambda\coloneqq \max\left\{ \frac{2}{\sqrt{d}}, \eps_0 \right\}$, where $\eps_0$ is the bias of the base code $\calC_0$.

\begin{remark}
    We remind the reader that we are working with codes in the $\{ \pm 1 \}$ basis rather than the standard $\{ 0, 1 \}$ basis. In the $\{ 0, 1 \}$ basis, the expander walk lift code would have codewords $\widehat{f}(\widehat{x})_{(i_0, \dots, i_\ell)} = \widehat{x}_{i_0} \oplus \widehat{x}_{i_1} \oplus \cdots \oplus \widehat{x}_{i_\ell}$ for $\wh{x} \in \{0,1\}^{n_0}$.
\end{remark}

% Through the rest of this paper, it will be more convenient for us to work in the $\{ \pm 1 \}$ basis rather than the $\{ 0, 1 \}$ basis. To this end, let us denote $\chi(x) = (-1)^x \in \{\pm 1)^{n_0}$, and  $f(\chi(x)) := \widehat{f(x)} = (-1)^{f(x)} \in \{ 0, 1 \}^I$. Then, 
% \[
%     f(\chi(x))_{i_0,\dots,i_\ell} = \chi(x)_{i_0} \cdot \chi(x)_{i_1} \cdots \chi(x)_{i_\ell}.
% \]

% The reason for working in the $\{ \pm 1 \}$ basis can be explained as follows. We'd like to bound the maximum bias of any codeword in $\calC$. Via a standard calculation, we can see that the relative weight of $f(x)$ is equal to $\frac{1 - \bbE_{i \in I} f(\chi(x))_i)}{2}$, so the bias of $f(x)$ is equal to $|\bbE_{i \in I} f(\chi(x))_i)|$. Thus, in order to show that $\calC$ is $\eps$-balanced, it suffices to upper bound $|\bbE_{i \in I} f(\chi(x))_i|$ by $\eps$ for all $x \in \calC_0$.

\subsection{Analyzing the Construction}
\label{sec:analysis}
Consider $x \in \calC_0 \backslash \{ \mathbf{1} \}$.
We can write: 
\begin{equation}
\label{eq:bias-formula}
    % \bias(f(x)) = \left| \big\langle \mathbf{1}, D_x A_{G_\ell} D_x A_{G_{\ell-1}} \cdots D_x A_{G_1} D_x \mathbf{1}\big\rangle \right|,
    \bias(f(x)) = \left| \big\langle \mathbf{1}, D_x A_{H_\ell} D_x A_{H_{\ell-1}} \cdots D_x A_{H_1} D_x \mathbf{1}\big\rangle \right|,
\end{equation}
where $D_x = \diag(x)$ and $A_{H_s} = \frac1d \Adj(H_s)$ where $\Adj(H_s) \in \{0,1\}^{n_0 \times n_0}$ denotes the adjacency matrix of $H_s$.
% Embed $f(x)$ into $\{\pm1\}^{n_0 d^\ell}$ by $\widehat{f(x)}=(-1)^{f(x)}$.
% A standard calculation gives the bias as an inner product of the form
% \begin{equation}
% \label{eq:bias-formula}
%   \frac{1}{n_0}\,\big\langle \mathbf{1},\; D_x A_{G_\ell} D_x A_{G_{\ell-1}} \cdots D_x A_{G_1} D_x \,\mathbf{1}\big\rangle.
% \end{equation}
Indeed, starting from the uniform distribution on $[n_0]$, each application of $D_x$ multiplies the current coordinate by $x_i$,
and each $A_{H_s}$ takes a uniform step to a neighbor in $H_s$.

We now introduce some notation to streamline the subsequent proofs.
\begin{itemize}
    \item We use $\Pi_{\mathbf{1}}$ to denote the projection onto the $\mathbf{1}$ direction, and $\Pi_{\perp}$ for the projection onto the subspace orthogonal to $\mathbf{1}$.
    \item We use $A_{j}$ as shorthand for $A_{H_j}$.
    \item Given a sequence of matrices $Q_1,\dots,Q_k$, we will use $\prod_{i=1}^k Q_i$ to denote $Q_k \cdots Q_1$.
    \item Given a word $w\in[t]^{\kappa}$, we define $M_w \coloneqq A_{w_{\arr}} D_x \cdots A_{w_1} D_x$.
\end{itemize}
With this notation in hand, we may write
\[
    \bias(f(x)) = \abs*{ \big\langle \mathbf{1}, D_x M_{w^{(r\cdot t^{\arr})}} \cdots M_{w^{(1)}} \mathbf{1} \big\rangle } \le \prod_{j=1}^{r\cdot t^{\arr}} \norm*{ M_{w^{(j)}} }_{\op} \mper   \numberthis \label{eq:prod-of-norm-bound}
\]
The key technical lemma is the following.
\begin{lemma}   \label{lem:random-sequence-good}
    Let $\bw$ be a uniformly random word drawn from $[t]^{\arr}$, where $t = d^2$.
    Then, for any $x \in \{\pm1\}^{n_0}$ with $\bias(x) \leq \eps_0$, with probability at least $1-\frac{\arr^2}{t^{1/4}}$, we have $\norm*{ M_{\bw} }_{\op} \le 2^{\arr+1} \lambda^{\arr-1}$. Recall that $M_w$ implicitly depends on $x$.
\end{lemma}

For the remaining $\le \frac{\arr^2}{t^{1/4}}$ fraction of the $w \in [t]^\arr$ not satisfying the above bound on $\norm{M_w}_{\op}$, we use the following trivial bound.

\begin{observation} \label{lem:triv-bound}
    For every $w\in[t]^{\arr}$, we have $\norm*{M_w}_{\op} \le 1$.
\end{observation}
\begin{proof}
    This is immediate from combining sub-multiplicativity of the operator norm with the fact that $M_w$ is a product of matrices with operator norm at most $1$.
\end{proof}

As an immediate corollary of \Cref{lem:random-sequence-good} and \Cref{lem:triv-bound}, we obtain the following.
\begin{lemma}   \label{lem:bias-bound}
    Let $x\in\calC_0\setminus\{\mathbf{1}\}$ of bias $\eps_0$.
    The bias of $f(x)$ is at most
    \[
        \lambda^{ \ell \cdot \parens*{1 - \frac{\arr^2}{t^{1/4}} - \frac{1}{\arr} } } \cdot 2^{\ell\cdot\frac{\arr+1}{\arr}}\mper
    \]
\end{lemma}

With \Cref{lem:bias-bound} in hand, we are ready to prove \Cref{thm:main}.
\begin{proof}[Proof of \Cref{thm:main}]
    Consider the concrete parameter setting in terms of some $\eps' > 0$:
    \begin{align*}
        d = \ceil*{\exp \parens*{\sqrt{\log \log \frac{1}{\eps'}}}} \,, \quad
        t = d^2 \,, \quad
        \arr = \floor*{\frac{1}{8} \sqrt{\log \log \frac{1}{\eps'}}} \,, \quad
        r = \left\lceil{ \frac{2\log_d\frac{1}{\eps'}}{\arr t^{\arr}} }\right\rceil \,, \quad
        \eps_0 = \frac{1}{\sqrt{d}} \,, \quad 
        % d = \left\lceil{\parens*{ \log\log\frac{1}{\eps'} }^8}\right\rceil, \quad 
        % t = d^2, \quad 
        % \arr = \left\lfloor\frac{\log\log\frac{1}{\eps'}}{32\log\log\log\frac{1}{\eps'}}\right\rfloor, \quad 
        % \eps_0 = \frac{1}{\sqrt{d}}, \quad 
        % r = \left\lceil{ \frac{2\log_d\frac{1}{\eps'}}{\arr t^{\arr}} }\right\rceil ~,
    \end{align*}
    so that $\lambda = \max \{ \frac{2}{\sqrt{d}}, \eps_0 \} = \frac{2}{\sqrt{d}}$.
    For simplicity, we will use $o(1)$ to denote any term at most $\frac{O(1)}{\sqrt{\log \log (1/\eps')}}$.
    With this notation, $\frac{1}{\log d}$ and $\frac{1}{\arr}$ are both $o(1)$.

    We first verify that $r \geq 1/o(1)$.
    Note that $\arr t^{\arr} \log d = \arr e^{2\arr \log d} \log d \leq (\log \frac{1}{\eps'})^{1/4} \cdot O(\log \log \frac{1}{\eps'})$.
    Thus, we have that $r = \ceil*{\frac{2\log (1/\eps')}{\arr t^{\arr} \log d}} \geq 1/o(1)$.
    Next, observe that $r$ is set such that $\ell = \arr t^{\arr} r = 2 \log_d(1/\eps') \cdot (1+o(1))$ (the error term is because of the ceiling), thus $d^{-\frac{1}{2}\ell} = (\eps')^{1-o(1)}$.
    
    From \cref{lem:bias-bound}, we have that for all $x \in \calC_0 \setminus \{\mathbf{1}\}$,
    \begin{align*}
        \bias(f(x))
        &\le \lambda^{ \ell \cdot \parens*{1 - \frac{\arr^2}{t^{1/4}} - \frac{1}{\arr} } } \cdot 2^{\ell\cdot\frac{\arr+1}{\arr}}
        \le 2^{2\ell} \cdot d^{-\frac{1}{2}\ell \parens*{1 - \frac{\arr^2}{t^{1/4}} - \frac{1}{\arr} } } \mper
    \end{align*}
    Both $\frac{\arr^2}{t^{1/4}}$ and $\frac{1}{\arr} \leq o(1)$,
    and $2^{2\ell} = (1/\eps')^{O(1/\log d)} = (\eps')^{-o(1)}$.
    Thus, the bias of our code is $(\eps')^{1-o(1)}$.

    On the other hand, the rate of the code is $\Omega(\eps_0^4) \cdot d^{-\ell}$, where $\Omega(\eps_0^4)$ is the rate of the base code $\calC_0$, and $d^{\ell}$ is the total number of length-$\ell$ walks.
    Then, $\eps_0^4 d^{-\ell} = d^{-\ell(1+\frac{2}{\ell})} = (\eps')^{2+o(1)}$.

    The theorem statement follows by choosing $\eps' = \eps^{1+o(1)}$ for a suitably chosen $o(1)$ term.
    The $o(1)$ in the statement of \Cref{thm:main} is precisely $\frac{O(1)}{\sqrt{\log \log (1/\eps)}}$.
\end{proof}

\subsection{Proof of \texorpdfstring{\Cref{lem:random-sequence-good}}{Lemma~\ref{lem:random-sequence-good}}}

It now remains to prove \Cref{lem:random-sequence-good}.
In service of doing so, we introduce the notion of a \emph{contractive sequence}.
\begin{definition}[Contractive sequence]
    We say that $w = (w_1, \dots, w_{\arr}) \in [t]^\arr$ is a \emph{contractive sequence} if for every $1\le r < s \le \arr$, we have:
    % \rnote{def 1: DA version}
    % \[
    %     \abs*{\langle u_1, D_x A_{w_{s}} \Pi_{u_1^\perp} D_x A_{w_{s-2}} \cdots \Pi_{u_1^\perp} D_x A_{w_r+1} \Pi_{u_1^\perp} D_x u_1 \rangle} \le \lambda^{s-r+1}.
    % \]
    % \rnote{def 2: AD version; preferred}
    \[
        \abs*{\langle \mathbf{1}, D_x \Pi_{\perp} A_{w_{s-1}} D_x \Pi_{\perp} A_{w_{s-2}}\dots  D_x \Pi_{\perp} A_{w_r} D_x \mathbf{1} \rangle} \le \lambda^{s-r+1}.
    \]
\end{definition}
We remark that there are $s-r$ copies of $A_{w_i}$ in the above equation, and thus the naive bound from second eigenvalue considerations alone is a bound of $\le \lambda^{s-r}$. Contractive sequences shrink more than the naive bound by a factor of $\lambda$.

\begin{observation}
    Every nonempty subsequence of a contractive sequence $w$ is contractive.
\end{observation}

\Cref{lem:random-sequence-good} then readily follows from \Cref{lem:random-seq-contractive,lem:contractive-norm-bound} below.
\begin{lemma}   \label{lem:random-seq-contractive}
    Let $\bw\sim[t]^{\arr}$ be drawn uniformly at random.
    For $t = d^2$, with probability $1 - \frac{\arr^2}{t^{1/4}}$, $\bw$ is a contractive sequence.
\end{lemma}

\begin{lemma}   \label{lem:contractive-norm-bound}
    Let $w$ be a contractive sequence. Then it holds that 
    $\norm{M_w}_{\mathrm{op}} \le 2^{\arr+1} \lambda^{\arr-1}$.
\end{lemma}

\subsubsection*{Proof of \Cref{lem:random-seq-contractive}: Most sequences are contractive}

\begin{proof}[Proof of \Cref{lem:random-seq-contractive}]
    Note that for any $w\in[t]$, $A_w$ commutes with $\Pi_{\perp}$, and hence the claim is equivalent to proving:
    \[
        \abs*{ \angles*{ \mathbf{1}, D_x A_{w_{s-1}} \Pi_{\perp} D_x A_{w_{s-2}} \Pi_{\perp} \cdots D_x A_{w_r} \Pi_{\perp} D_x \mathbf{1}} } \le \lambda^{s-r+1}\mper
    \]
    We will first prove that for any $y$ such that $\angles*{y,\mathbf{1}} = 0$, we have $\E_{\bu\sim[t]} \abs*{ \angles*{\mathbf{1}, D_x A_{\bu} y} } \le \frac{2}{\sqrt{td}} \norm*{y} $.
    We may prove this by expanding:
    \begin{align*}
        \E_{\bu\sim[t]} \abs*{ \angles*{ \mathbf{1}, D_x A_{\bu} y } } &= \E_{\bu\sim[t]} \sgn\parens*{ \angles*{ \mathbf{1}, D_x A_{\bu} y } } \angles*{ \mathbf{1}, D_x A_{\bu} y } \\
        &= \angles*{ D_x \mathbf{1},\ \E_{\bu\sim[t]} \bracks*{\sgn\parens*{ \angles*{ \mathbf{1}, D_x A_{\bu} y } } A_{\bu}} \cdot y } \\
        &\le \frac{2}{\sqrt{td}} \norm*{y}\mcom
    \end{align*}
    where the last inequality uses the all-signings near-Ramanujan property.
    By Markov's inequality, with probability at least $1-t^{-1/4}$, for $\bu\sim[t]$ we have:
    \[
        \abs*{ \angles*{ \mathbf{1}, D_x A_{\bu} y } } \le \frac{2}{t^{1/4}\sqrt{d}} \norm*{ y } \le \lambda^2 \norm*{y} \mper
    \]
    where the last inequality used the concrete choice of $\lambda$ and $t = d^2$.

    Using the fact that $\norm*{ A_{H_{w_i}}|_{\mathbf{1}^{\perp}} }_{\op} \le \frac{2}{\sqrt{d}} \le \lambda$, with probability at least $1-\frac{1}{t^{1/4}}$:
    \begin{align*}
        \abs*{\angles*{ \mathbf{1}, D_x A_{\bw_{s-1}} \Pi_{\perp} D_x A_{\bw_{s-2}} \Pi_{\perp} \cdots D_x A_{\bw_r} \Pi_{\perp} D_x \mathbf{1} }} &\le \lambda^2 \norm*{D_x A_{w_{s-2}} \Pi_{\perp} \cdots A_{w_r} \Pi_{\perp} D_x \mathbf{1} } \\
        &\le \lambda^2 \cdot \lambda^{s-r-1} \\
        &\le \lambda^{s-r+1}\mper
    \end{align*}
    Taking a union bound over all choices of $r$ and $s$ implies that $\bw$ is a contractive sequence with probability at least $1 - \frac{\arr^2}{t^{1/4}}$.
\end{proof}

\subsubsection*{Proof of \Cref{lem:contractive-norm-bound}: Bounds for operator norm of contractive sequences}

The proof proceeds by separately analyzing the effect of $M_w$ on $\mathbf{1}$ and $\mathbf{1}^{\perp}$.

\begin{lemma}\label{lem:1-direction-bound} % \rnote{AD split}
    For any nonnegative integer $\arr$, and any contractive sequence $w\in[t]^{\arr}$, it holds that
    \[
        \norm{M_w \mathbf{1}} \le (2\lambda)^{\arr}\mper
    \]
\end{lemma}

\begin{proof}
    We will prove this lemma by induction on $\arr$. For $\arr = 0$, the statement is clear: $\norm{\mathbf{1}} = 1 = (2\lambda)^0$. Next, for $\arr \ge 1$, by iteratively splitting by projection onto $\mathbf{1}$ and $\mathbf{1}^\perp$ after each step, we obtain the following decomposition:
    \begin{align*}
        \left( \prod_{j=1}^{\arr} A_j D_x \right) \mathbf{1}
        =&~ \left( \prod_{j=1}^{\arr} \Pi_{\perp} A_j D_x \right) \mathbf{1}
        + \sum_{s=1}^{\arr}~ \left( \prod_{j=s+1}^{\arr} A_j D_x \right)\Bigg( \Pi_{\mathbf{1}} A_s D_x \Bigg) \left( \prod_{j=1}^{s-1} \Pi_{\perp} A_j D_x \right) \mathbf{1}\mper
    \end{align*}
    Thus, by triangle inequality, we have that 
    \begin{align*}
        \norm*{\left( \prod_{j=1}^{\arr} A_j D_x \right) \mathbf{1}}
        \le&~ \norm*{\left( \prod_{j=1}^{\arr} \Pi_{\perp} A_j D_x \right) \mathbf{1} } \\
        &+ \sum_{s=1}^{\arr}~ \norm*{\left( \prod_{j=s+1}^{\arr} A_j D_x \right)\Bigg( \Pi_{\mathbf{1}} A_s D_x \Bigg) \left( \prod_{j=1}^{s-1} \Pi_{\perp} A_j D_x \right) \mathbf{1}}\mper \numberthis \label{eqn:1-direction-recursion}
    \end{align*}
    Let us analyze each component of the right hand side of this expression. 
    \begin{itemize}
    \item 
        First we analyze $\norm*{\left( \prod_{j=1}^{\arr} \Pi_{\perp} A_j D_x \right) \mathbf{1}}$.
        For any vector $u$ and for any $j \in [\arr]$, we have that $\Pi_{\perp} A_j D_x u$ is a vector perpendicular to $\mathbf{1}$ and one can bound its norm as follows:
        \[
            \norm{\Pi_{\perp} A_j D_x u} 
            = \norm{A_j \Pi_{\perp} D_x u} 
            \le \lambda \norm{\Pi_{\perp} D_x u}
            \le \lambda \norm{D_x u}
            \le \lambda \norm{u}\mcom
        \]
        where the first equality holds because $\Pi_{\perp}$ and $A_j$ commute since $\mathbf{1}$ is an eigenvector of $A_j$.
        Thus, inductively applying $\Pi_{\perp} A_j D_x$ to $\mathbf{1}$, we have that 
        \[
            \norm*{\left( \prod_{j=1}^{\arr} \Pi_{\perp} A_j D_x \right) \mathbf{1}} 
            \le \lambda^{\arr}\mper
        \]
    \item 
        Next, let us consider $\norm*{\left( \prod_{j=s+1}^{\arr} A_j D_x \right)\Big( \Pi_{\mathbf{1}} A_s D_x \Big) \left( \prod_{j=1}^{s-1} \Pi_{\perp} A_j D_x \right) \mathbf{1}}$ for each $s \in [\arr]$. Since $\Pi_{\mathbf{1}} A_s = \Pi_{\mathbf{1}}$, this is also equal to
        \begin{align*}
            \Bigg\|\left( \prod_{j=s+1}^{\arr} A_j D_x \right) & \Bigg( \Pi_{\mathbf{1}} D_x \Bigg) \left( \prod_{j=1}^{s-1} \Pi_{\perp} A_j D_x \right) \mathbf{1} \Bigg\| \\
            % &= \left( \prod_{j=s+1}^a A_j D_x \right) u_1 \cdot \norm*{\Bigg( \Pi_{u_1} D_x \Bigg) \left( \prod_{j=1}^{s-1} \Pi_{u_1^\perp} A_j D_x \right) u_1} \\
            &= \norm*{\left( \prod_{j=s+1}^{\arr} A_j D_x \right) \mathbf{1}} \cdot \abs*{\left\langle \mathbf{1}, D_x \left( \prod_{j=1}^{s-1} \Pi_{\perp} A_j D_x \right) \mathbf{1} \right\rangle} \\
            &\le \norm*{\left( \prod_{j=s+1}^{\arr} A_j D_x \right) \mathbf{1} } \cdot \lambda^s \\
            &\le (2\lambda)^{\arr-s} \cdot \lambda^s\mper
        \end{align*}
        In the above, the penultimate inequality follows from the fact that $w_1, \dots, w_{s-1}$ is a contractive sequence for $s \ge 2$, and by the fact that for $s = 1$, $\abs*{\angles*{ \mathbf{1}, D_x \mathbf{1} }} \le \eps_0 \le \lambda$.
        The final inequality is from strong induction on the subsequence $w_{s+1},\dots,w_{\arr}$ also being a contractive sequence.
        % This gives us that its norm is bounded by
        % \begin{align*}
        %     \norm*{\left( \prod_{j=s+1}^a A_j D_x \right) \Bigg( \Pi_{u_1} D_x \Bigg) \left( \prod_{j=1}^{s-1} \Pi_{u_1^\perp} A_j D_x \right) u_1}
        %     &\le \lambda^s \cdot \norm*{\left( \prod_{j=s+1}^a A_j D_x \right) u_1} \\
        %     &\le \lambda^s \cdot (2\lambda)^{a-s} \\
        %     &= 2^{a-s} \lambda^a,
        % \end{align*}
        % where the second inequality is from strong induction on the subsequence $w_{s+1}, \dots, w_a$ also being a contractive sequence.
    \end{itemize}
    Now, returning to \Cref{eqn:1-direction-recursion}, we get that
    \[
        \norm*{\left( \prod_{j=1}^{\arr} A_j D_x \right) \mathbf{1}} 
        \le \lambda^{\arr} + \sum_{s=1}^{\arr} 2^{\arr-s}\lambda^{\arr}
        = (2\lambda)^{\arr}\mper    \qedhere
    \]
\end{proof}

\begin{lemma}\label{lem:perp-direction-bound}
    For any contractive sequence $w$ and for any $\vee \perp \mathbf{1}$, it holds that
    \[
        \norm*{M_w \vee} \le 2^{\kappa} \lambda^{\kappa-1} \norm{\vee}\mper
    \]
\end{lemma}

\begin{proof}
    We will prove this claim by induction on $\arr$.
    For the base case of $\arr=0$, the statement follows from $\norm{\vee} \le \frac{1}{\lambda} \cdot \norm{\vee}$.

    For $\arr \ge 1$, let $\vee_1 = \Pi_{\mathbf{1}} A_1 D_x \vee$ and $\vee' = \Pi_{\perp} A_1 D_x \vee$. These satisfy that $\norm{\vee_1} \le \norm{\vee}$ and $\norm{\vee'} = \norm{\Pi_{\perp} A_1 D_x \vee} = \norm{A_1 \Pi_{\perp} D_x \vee} \le \lambda \norm{\vee}$, where we use that $A_1$ and $\Pi_{\perp}$ commute since $\mathbf{1}$ is an eigenspace of $A_1$.
    Then,
    \begin{align*}
        \norm*{ \left(\prod_{j=1}^{\arr} A_j D_x \right) \vee}
        &\le \norm*{ \left(\prod_{j=2}^{\arr} A_j D_x \right) \vee_1}
        + \norm*{ \left(\prod_{j=2}^{\arr} A_j D_x \right) \vee'} \\
        &\le (2\lambda)^{\arr-1} \norm{v_1} + 2^{\arr-1} \lambda^{\arr-2} \norm{v'} \\
        &\le (2\lambda)^{\arr-1} \norm{v} + 2^{\arr-1} \lambda^{\arr-2} \cdot \lambda \norm{v} \\
        &= 2^{\arr} \lambda^{\arr-1} \norm{v},
    \end{align*}
    where we use in the second inequality \Cref{lem:1-direction-bound} for the contractive subsequence $(w_2, \dots, w_{\arr})$ to bound $\norm*{ \left(\prod_{j=2}^{\arr} A_j D_x \right) \vee_1}$ and the inductive hypothesis for $(w_2, \dots, w_{\arr})$ to bound $\norm*{ \left(\prod_{j=2}^a A_j D_x \right) \vee'}$. 
\end{proof}

\begin{proof}[Proof of \Cref{lem:contractive-norm-bound}]
    Let $\vee$ be any vector in $\R^{n_0}$, and let $\vee_1 = \Pi_{\mathbf{1}} \vee$ and $\vee' = \Pi_{\perp} \vee$.
    Combining the bounds from \Cref{lem:1-direction-bound,lem:perp-direction-bound}, we obtain that
    \begin{align*}
        \norm*{M_w \vee} &
        \le \norm*{M_w \vee_1} 
        + \norm*{M_w \vee'}
        \le (2\lambda)^\arr \norm{\vee_1} + 2^\arr \lambda^{\arr-1} \norm{\vee'}
        \le 2^{\arr+1} \lambda^{\arr-1} \norm{\vee}\mper \qedhere
    \end{align*}
\end{proof}

%% file: free-expanders.tex
\section{Free Expanders}
\label{sec:free-expanders}

In our construction of $\eps$-balanced codes, we use explicit constructions of ``free expanders'' given by O'Donnell and Wu \cite{OW20} that satisfy the \emph{all-signings near-Ramanujan property} (\Cref{def:all-signings-expanders}).
However, their construction is significantly more general.
In this section, we describe their result, highlight some additional properties and show several applications of their result.

To describe their construction, we will need the notion of a \emph{noncommutative polynomial} and a \emph{free group}.
Let $\calW_k$ denote the set of all words of length $\le k$ on self-adjoint indeterminates $Y_1,\dots,Y_D$, i.e. $Y_i=Y_i^*$.
A noncommutative polynomial of degree $\le k$ on these indeterminates can generally be written as:
\[
    P(Y_1,\dots,Y_D) = \sum_{W\in\calW_k} \wh{P}_W W\mcom
\]
for $\wh{P}_W\in\R$.
We say that a polynomial is \emph{self-adjoint} is $P^* = P$, and define the \emph{Frobenius norm} of a polynomial as $\norm*{P}_F = \sqrt{ \sum_{W\in\calW_k} |\wh{P}_W|^2 }$.

Let $F_D$ be the \emph{free group} on $D$ self-adjoint generators $x_1,\dots,x_D$ such that $x_i^2 = 1$.
We identify elements of the free group via their \emph{adjacency operators}---formally, we identify $g\in F_D$ as a linear operator on $\ell_2(F_D) \coloneqq \braces*{ f:F_D\to\R \mid \sum_{g\in F_D} f(g)^2 < \infty }$ where $gf(u) = f(gu)$.

A key result of interest from \cite{OW20} is on the existence of an infinite family of finite matching matrices that spectrally approximate the free group for \emph{all} low-degree polynomials.
\begin{theorem}[{\cite[Special case of Theorem 10.13]{OW20}}]   \label{thm:OW-main}
    For every $D, k \ge 1$, and $r,\eps > 0$, there is an algorithm that takes in a positive integer input $N$, and in $\poly(N)$ time outputs perfect matching matrices $X_1,\dots,X_D$ of dimension $N'\times N'$ where $N' = N + o(N)$, such that simultaneously for every degree-$\le k$ self-adjoint polynomial $P(Y_1,\dots,Y_D)$ such that $\norm*{P}_F \le r$, we have:
    \[
        \spec\parens*{p(X_1,\dots,X_D)\big|_{\mathbf{1}^{\perp}}} \subseteq \spec\parens*{p(x_1,\dots,x_D)} \pm [-\eps,\eps]\mper
    \]
\end{theorem}

The result of \cite{OW20} also applies to polynomials with \emph{matrix-valued} coefficients, which is relevant to constructing expanding lifts of fixed constant-sized expander graphs.

\begin{remark}[Near-Ramanujan graphs]
    Observe that choosing $P(X_1,\dots,X_D) = \sum_{i=1}^D X_i$ yields a $D$-regular graph whose spectrum approximates that of $\sum_{i=1}^D x_i$, which is the adjacency operator of the $D$-regular infinite tree whose spectrum is $[-2\sqrt{D-1}, 2\sqrt{D-1}]$.
    Therefore, the construction yields a near-Ramanujan graph.
\end{remark}

\begin{remark}[All-signings near-Ramanujan family]
    We may obtain \Cref{thm:OW} as a corollary of \Cref{thm:OW-main} by choosing $D = dt$, and restricting our attention to the collection of signed linear combinations of $X_1,\dots,X_D$.
    The graphs $H_1,\dots,H_t$ can be chosen by picking $H_i \coloneqq X_{d(i-1)+1} + \dots + X_{di}$.
    The spectrum of any signed linear combination of $x_1,\dots,x_D$ is $[-2\sqrt{D-1}, 2\sqrt{D-1}]$.
\end{remark}

\subsection{Lossless Expanders on Average}

We say that a $d$-regular graph $G = (V, E)$ is a lossless expander if every subset $S \subseteq V$ of size at most $\eta |V|$ (where $\eta$ depends only on $d$) has at least $(1-\eps) d |S|$ distinct neighbors, with arbitrarily small $\eps = \eps(d) \to 0$ as $d \to \infty$.
The explicit construction of constant-degree lossless expanders has been a long-standing open question \cite{HLW06}, and this was recently resolved in a line of work \cite{HMMP24,HLMOZ25,HLMRZ25}.

Using \cite{OW20}, we can construct a collection of $d$-regular graphs $H_1,\dots,H_t$ on the same vertex set that are ``lossless expanders on average'', meaning that any small subset $S \subseteq V$ has lossless expansion in most graphs among $H_1,\dots,H_t$.
The specific property we need is stated below, which follows from \Cref{thm:OW-main} and \cite[{\S}2.2, Eq. (10)]{EFS17}.

\begin{theorem} \label{thm:OW-2-step-walk}
    For every integer $t \geq 1$ and $d \geq 2$, there is an explicit construction of $d$-regular graphs $H_1, \dots, H_t$ on the same vertex set such that
    \begin{align*}
        \max_{\sigma \in \{\pm1\}^t} \norm*{\frac{1}{t} \sum_{j=1}^t \sigma_j \parens*{A_j^2 - \frac{1}{d} \Id}^{\perp} }_{\op}
        \leq
        O\parens*{\frac{1}{\sqrt{t} \cdot d}} \mcom
    \end{align*}
    where $A_j = \frac{1}{d} \Adj(H_j)$ is the normalized adjacency matrix of $H_j$, and $A_j^\perp$ is the matrix with the $\1$ eigenspace removed.
\end{theorem}

Intuitively, $A_j^2$ is the transition matrix of $2$-step walks on $H_j$, which has a $\frac{1}{d}$ probability of going back to the starting vertex.
Thus, $A_j^2 - \frac{1}{d} \Id$ is the transition matrix of $2$-step \emph{non-backtracking} walks on $H_j$, and \Cref{thm:OW-2-step-walk} can be viewed as the all-signings near-Ramanujan property of such walks.

With \Cref{thm:OW-2-step-walk}, we show the following.

\begin{theorem}[Lossless expanders on average]
    Let $d,t \in \N$.
    There is an explicit construction of $d$-regular graphs $H_1,\dots,H_t$ on vertex set $V$ that satisfy the following:
    let $\eps,\eta,\gamma \in (0,1)$ such that $\eta \leq O(\gamma \eps /d)$ and $t \geq \Omega(\frac{1}{\gamma^2 \eps^2})$, then for any subset $S \subseteq V$ of size $|S| \leq \eta |V|$, at least $1-\gamma$ fraction of $j\in [t]$ has $|N_j(S)| \geq (1-\eps) d|S|$.
\end{theorem}

\begin{proof}
    We will show that the graphs $H_1,\dots,H_t$ from \Cref{thm:OW-2-step-walk} satisfy lossless expansion on average.
    Fix a subset $S \subseteq V$ of size at most $\eta n$.
    Consider the quadratic form $\1_S^\top A_j^2 \1_S$, where $\1_S$ is the 0-1 indicator vector of $S$.
    Since $A_j$ is the normalized adjacency matrix of $H_j$, $\1_S^\top A_j^2 \1_S$ equals $\frac{1}{d^2}$ times the number of length-$2$ walks that start and end in $S$.
    Since a length-$2$ walk can be described by choosing the middle vertex $v$, and picking two (possibly identical) neighbors of $v$ in $S$, we may write the number of length-$2$ walks that start and end in $S$ as $\sum_{v\in N_j(S)} d_{j,v}^2$, where $d_{j,v} = |E_{H_j}(v, S)|$.
    Thus, we may write
    \begin{align*}
        \1_S^\top A_j^2 \1_S = \frac{1}{d^2} \sum_{v \in N_j(S)} d_{j,v}^2 = \frac{|N_j(S)|}{d^2} \cdot \E_{v\in N_j(S)}[d_{j,v}^2] \ge \frac{|N_j(S)|}{d^2} \cdot \E_{v\sim N_j(S)}[d_{j,v}]^2 \mcom 
    \end{align*}
    where $N_j(S)$ refers to the neighborhood of $S$ in $H_j$.
    Observe that $\sum_{v \in N_j(S)} d_{j,v}$ is equal to the number of edges from $S$ to $N_j(S)$, which is equal to $d|S|$.
    Thus, $\E_{v\sim N_j(S)}[d_{j,v}] = \frac{d|S|}{|N_j(S)|}$, and we obtain
    \begin{align*}
        \1_S^\top A_j^2 \1_S \geq \frac{|N_j(S)|}{d^2} \cdot \frac{d^2 |S|^2}{|N_j(S)|^2}
        = \frac{|S|^2}{|N_j(S)|} \mper 
    \end{align*}
    On the other hand, writing $\1_S = \frac{|S|}{n} \1 + \1_S^\perp$ where $\1_S^\perp$ is orthogonal to $\1$, we have $\1_S^\top A_j^2 \1_S = \frac{|S|^2}{n} + \1_S^\top (A_j^2)^\perp \1_S$.
    Then,
    \begin{align*}
        \frac{1}{t} \sum_{j=1}^t \abs*{\1_S^\top \parens*{A_j^2 - \frac{1}{d} \Id} \1_S} 
        &\leq \frac{|S|^2}{n} + \max_{\sigma \in \{\pm1\}^t} \frac{1}{t} \sum_{j=1}^t \sigma_j \cdot \1_S^\top \parens*{A_j^2 - \frac{1}{d} }^{\perp} \1_S \\
        &\leq \eta |S| + |S| \cdot \max_{\sigma \in \{\pm1\}^t} \norm*{\frac{1}{t} \sum_{j=1}^t \sigma_j \parens*{A_j^2 - \frac{1}{d} \Id}^{\perp} }_{\op} \\
        &\leq |S| \cdot \parens*{\eta + O\parens*{\frac{1}{\sqrt{t} d}}} \mcom
    \end{align*}
    where the last inequality follows from the guarantee of \Cref{thm:OW-2-step-walk}.
    By Markov's inequality, at least $1-\gamma$ fraction of $j \in [t]$ satisfies
    \begin{align*}
        \abs*{\1_S^\top \parens*{A_j^2 - \frac{1}{d} \Id} \1_S} \leq \frac{|S|}{d} \cdot \parens*{\frac{\eta d}{\gamma} + O\parens*{\frac{1}{\gamma \sqrt{t}}}}
        \leq \frac{\eps |S|}{d} \mcom
    \end{align*}
    as long as $\eta \leq O(\gamma \eps/d)$ and $t \geq \Omega(\frac{1}{\gamma^2 \eps^2})$.
    Combined with the lower bound $\1_S^\top A_j^2 \1_S \geq \frac{|S|^2}{|N_j(S)|}$, we have that $1-\gamma$ fraction of $j\in [t]$ satisfies
    \begin{equation*}
        \frac{|S|^2}{|N_j(S)|} \leq \frac{|S|}{d} (1+\eps) \implies |N_j(S)| \geq (1-\eps) d|S| \mper
        \qedhere
    \end{equation*}
\end{proof}

\subsection{Explicit Rotating Expanders}

One of the sources of suboptimality in our construction as well as previous work \cite{TaS17,CohenC25} is the factor of $2$ per step on a Ramanujan graph.
If one were to obtain an explicit construction of $\eps$-balanced codes achieving the GV bound without an $\eps^{o(1)}$ overhead in the rate based on expander walks, it would be necessary to circumvent the fact that the operator norm of the adjacency matrix of any Ramanujan graph has spectral radius $\approx\frac{2}{\sqrt{d}}$, as opposed to $\frac{1}{\sqrt{d}}$.

In the context of mixing of random walks, it has been observed that this factor of $2$ can be shaved if one performs a nonbacktracking walk on a Ramanujan graph \cite{ABLS07}.
Cohen and Maor \cite{CM23} used the method of interlacing polynomials and finite free probability (cf.\ \cite{MSS18}) to exhibit the existence of ``rotating expanders'':
a sequence of $d$-regular (one-sided) Ramanujan graphs $G_1,\dots,G_{\ell}$ such that the second eigenvalue of $G_1\cdot G_2 \cdots G_{\ell} \cdot G_{\ell} \cdot G_{\ell-1} \cdots G_1$ is at most $O(\ell) \cdot d^{-\ell}$, as opposed to $2^{2\ell} d^{-\ell}$.
Moreover, they can guarantee that each $G_i$ is a relabeling of $G_1$.

We observe below that an explicit sequence of expanders satisfying this property can be obtained from \Cref{thm:OW-main}.
In particular, by combining \Cref{thm:OW-main} and a result of Haagerup \cite[Lemma 1.4]{Haa78}, we have the following.
\begin{theorem} \label{thm:rotating-expanders}
    Fix $d, \ell \in \N$.
    There is an explicit construction of $d$-regular graphs $H_1,\dots,H_t$ on the same vertex set such that each $H_j$ satisfies $\|A_{H_j}^\perp\|_{\op} \leq \frac{2}{\sqrt{d}}$, and moreover,
    \begin{align*}
        \norm*{\parens*{A_{H_\ell} A_{H_{\ell-1}} \cdots A_{H_1}}^\perp }_{\op} \leq O(\ell) \cdot d^{-\ell/2} \mper
    \end{align*}
    Further, it is possible to realize $H_i$ as $P_i H_1 P_i^{\top}$ for a permutation matrix $P_i$.
\end{theorem}

Note that the naive bound in this case is $2^{\ell} d^{-\ell/2}$, and this would be tight if all the graphs are identical.
On the other hand, $A_{H_\ell} A_{H_{\ell-1}} \cdots A_{H_1}$ is the transition matrix of $\ell$-step walks on $d$-regular graphs, which can be viewed as the transition matrix of a $d^\ell$-regular graph.
If this graph were Ramanujan, we would get a bound of $2\cdot d^{-\ell/2}$.

In the context of $\eps$-balanced codes,
it remains unclear how the factor of $2^{\ell}$ can be removed in the setting of \Cref{eq:bias-formula}, where one inserts a diagonal matrix $D_x$ between consecutive $A_{H_i}$.
However, we hope that this phenomenon will one day lead to improved analysis of such expander walks, yielding explicit $\eps$-balanced codes closer to the GV bound.